\documentclass{comsoc2025}
\usepackage{amsmath,amsfonts,amssymb}
\usepackage{amsthm}

\usepackage{hyperref}
\usepackage{cleveref}
\usepackage{mathtools}
\usepackage[noend,ruled]{algorithm2e}
\usepackage{placeins}
\usepackage{natbib}
\usepackage{tikz}

\tikzstyle{voter}=[circle, draw, fill=white, minimum size=7mm, inner sep=1pt]
\tikzstyle{winner}=[circle, draw, fill=white, minimum size=7mm, inner sep=1pt]
\tikzstyle{terminal}=[circle, draw, fill=white, minimum size=7mm, inner sep=1pt]
\tikzstyle{source side}=[circle, draw, fill=blue!18, minimum size=7mm, inner sep=1pt]
\tikzstyle{sink side}=[circle, draw, fill=black!12, minimum size=7mm, inner sep=1pt]
\tikzstyle{flow edge}=[->, >=stealth, thick]
\tikzstyle{cut edge}=[->, >=stealth, very thick, red!75!black]
\tikzstyle{edge label}=[fill=white, inner sep=1pt, font=\scriptsize]

\crefname{algocf}{alg.}{algs.}
\Crefname{algocf}{Algorithm}{Algorithms}
\newtheorem{theorem}{Theorem}
\newtheorem{remark}[theorem]{Remark}
\newtheorem{lemma}[theorem]{Lemma}

\newtheorem{proposition}[theorem]{Proposition}

\newtheorem{definition}{Definition}

\newcommand{\Oh}{\mathcal{O}}

\DeclarePairedDelimiter\floor{\lfloor}{\rfloor}
\newcommand{\ignore}[1]{}
\title{Faster Verification of PJR$^+$ via Mincuts}
\author{Drew Springham\footnote{King's College London, drew.springham@kcl.ac.uk}}
\hypersetup{
  pdftitle={Faster Verification of PJR+ via Mincuts},
  pdfauthor={Drew Springham}
}

\begin{document}
%%%%%%%%%%%%%%%%%%%%%%%%%%%%%%%%%%%%%%%%%%%%%%%%%%%%%%%%%%%%%%%%%%%%%%%%%

% Include a short abstract here (100-300 words):
\begin{abstract}
PJR$^+$ is a polynomial-time verifiable proportionality axiom for
approval-based committee elections, but its known polynomial-time verification
procedure relies on general submodular-function minimisation. We show that its
objective is a maximum-closure
problem and give a direct mincut formulation of the problem on a bipartite graph. Using
an almost-linear-time maximum-flow algorithm, this yields an
\(\Oh(m(nk)^{1+o(1)})\)-time verifier, where \(n\), \(m\), and \(k\) are the
numbers of voters, candidates, and committee members, respectively. The
dependence of this bound on each parameter separately is almost linear: it is
linear in \(m\), and almost linear in \(n\) and \(k\). The
verifier also returns an explicit group witnessing a violation and admits
a slower but immediately implementable variant based on the preflow--push mincut algorithm. Finally,
for the parameterised axiom \(\alpha\)-PJR$^+$, where \(\alpha\) is used as a multiplier
in the group size, we demonstrate how to compute the largest value of \(\alpha\) for
which a committee still fails the axiom using this mincut formulation.
\end{abstract}

%%%%%%%%%%%%%%%%%%%%%%%%%%%%%%%%%%%%%%%%%%%%%%%%%%%%%%%%%%%%%%%%%%%%%%%%%

\section{Introduction}
Approval-based multiwinner voting selects a committee from candidates using voters'
approval ballots. Proportionality axioms aim to ensure that sufficiently large groups
with shared preferences receive an appropriate number of representatives.
Proportional justified representation (PJR) requires, for every positive integer
\(\ell\), that any group comprising at least an \(\ell/k\) fraction of the voters
and sharing at least \(\ell\) approved candidates collectively approve at least
\(\ell\) committee members
\citep{DBLP:journals/ai/SanchezFernandezELGFBS26}. PJR$^+$
\citep{DBLP:conf/sigecom/Brill023} strengthens this requirement by anchoring a
group's claim to a single commonly approved
unelected candidate\footnote{PJR$^+$ is the approval-based specialisation of inclusion
proportionality for solid coalitions (IPSC) axiom, introduced by \citet{DBLP:conf/aaai/0001L21}. \citet{DBLP:conf/sigecom/Brill023} 
reformulated this condition for approval-based elections and named it
PJR$^+$.}. Although PJR$^+$ is stronger, it can be verified in polynomial
time, whereas deciding whether a given committee satisfies PJR is
coNP-complete \citep{DBLP:conf/aaai/0001EHLFS18}. This note gives a faster
and more direct way to perform PJR$^+$ verification.

The landscape of efficiently verifiable proportionality axioms has recently
expanded beyond PJR$^+$ and EJR$^+$. FJR$^+$ is a fractional strengthening of
FJR and EJR$^+$ \citep{teh2026strengtheningjustifiedrepresentationefficient},
while core$^+$ is a fractional strengthening of the core and FJR$^+$
\citep{becker2026core}. Both axioms are polynomial-time verifiable, and
committees satisfying either axiom can also be computed in polynomial time.
These results reinforce the usefulness of efficient
verification even for demanding proportionality requirements. Our focus is on
the running time of verification: we replace the general
submodular-minimisation procedure for PJR$^+$ with an explicit mincut
formulation.

Scaling the minimum size of groups entitled to representation is a common way
to obtain approximate or quantitative variants of proportionality axioms.
To compare the different conventions, let \(\alpha\) denote a multiplier on
the usual group-size threshold: a group claiming \(\ell\) representatives must
contain at least \(\alpha\ell n/k\) voters.
In the incomplete-vote setting, \citet{halpern2026representation} define
\(\alpha\)-EJR using the reciprocal convention: their parameter divides,
rather than multiplies, the usual group-size threshold. For ordinal elections,
\citet{DBLP:conf/aaai/BardalBM025} apply group-size scaling to EJR$^+$ and
other axioms, including multipliers \(\alpha<1\). A fixed factor-two
relaxation of EJR$^+$ is also used by \citet{jianggoel2025approximation} when
optimising committee cost.

For PJR$^+$, the present group-size scaling with \(\alpha\geq1\) appeared in
\citet{springham26a}. This paper extends the definition to every \(\alpha>0\),
including \(\alpha<1\), and computes the exact boundary value for a given
committee. This threshold provides a quantitative comparison complementary to
empirical comparisons based on the prevalence, satisfaction, and coverage of
JR, PJR, and EJR committees \citep{DBLP:conf/ijcai/BredereckFKN19}.
Relatedly, \citet{he2026checkpleaseverifiablyfair} study ($\alpha=1$) PJR\(^{+}\)
verification in proportional clustering, introducing a metric analogue of
PJR\(^{+}\) and corresponding verification algorithms.

\citet{DBLP:conf/sigecom/Brill023} verify PJR$^+$ by solving one
submodular-minimisation problem for each unelected candidate. Their objective is the
number of committee members approved by at least one voter in a chosen group, minus
a linear reward for the size of that group. The key structure is that including
a voter gives a fixed reward, whereas each committee member approved by at
least one included voter incurs a cost only once, regardless of how many
included voters approve them. This is a classical maximum-closure problem
\citep{158279a7-9414-3656-aa58-e0056bf82d76}, equivalently a selection problem
with shared coverage costs
\citep{4528f944-06bd-31a8-80a9-a910bab75944,DBLP:journals/mansci/Hochbaum04}.
Both formulations admit a direct mincut representation in a bipartite
graph, which we exploit to obtain a direct verifier.

An exact mincut in a directed graph with \(b\) edges and
polynomially bounded integral capacities can be computed in
\(b^{1+o(1)}\) time
\citep{DBLP:conf/focs/Brand0PKLGSS23,DBLP:journals/jacm/ChenKLPGS25}.
This gives an almost-linear-time algorithm for PJR$^+$ verification, with
running time $\Oh\bigl(m(nk)^{1+o(1)}\bigr)$, where $n$ is the number of
voters, $m$ the number of candidates, and $k$ the committee size. This is
``almost-linear'' in each variable separately. The verifier also returns an
explicit group witnessing a violation. For comparison, applying the expected
weakly polynomial running-time bound of
\citet{DBLP:conf/focs/LeeSW15} to the scaled integer-valued objective gives a
total running time of
\(\Oh\bigl(mn^3k\log^{\Oh(1)}(nk)\bigr)\): the ground set has at most \(n\)
voters, a value-oracle call takes \(\Oh(nk)\) time, and the magnitude of the
objective is \(\Oh(nk)\). The direct mincut formulation therefore removes
essentially a quadratic factor in \(n\), up to subpolynomial and polylogarithmic
factors.

The almost-linear-time maximum-flow algorithms underlying this bound are currently
mostly theoretical; an overview of these results notes that they ``remain
impractical'' \citep{Herkle_2023}. We
therefore also analyse a specialised bipartite highest-label preflow--push
algorithm \citep{DBLP:journals/siamcomp/AhujaOST94} and a generic
highest-label preflow--push algorithm
\citep{DBLP:journals/jacm/GoldbergT88,DBLP:journals/siamcomp/CheriyanM89}. The latter is available in NetworkX \citep{hagberg2008}, providing
an off-the-shelf Python implementation of the mincut component\footnote{See \url{https://networkx.org/documentation/networkx-3.6.1/reference/algorithms/generated/networkx.algorithms.flow.preflow_push.html}}. This alternative has weaker worst-case bounds but makes the
reduction directly usable without implementing the recent almost-linear algorithms.

We additionally study $\alpha$-PJR$^+$, in which the size required of a potentially
underrepresented group is scaled by a parameter $\alpha>0$. This leads to the
\emph{PJR$^+$ threshold}: the largest value of $\alpha$ at which the committee still
fails the axiom. A smaller threshold therefore gives a stronger guarantee. We show the
threshold can be computed exactly by binary search over a finite set of possible
values, giving an immediately implementable method using ordinary mincut
routines.
For completeness, Appendix~\ref{app:parametric} records a specialised
parametric bipartite-flow formulation, whose worst-case running time is
effectively the same as that of binary search.

Our contributions are as follows:
\begin{itemize}
    \item We recognise the maximum-closure structure of the Brill--Peters
    verification objective and obtain a direct mincut verifier for PJR$^+$ that
    returns an explicit violation witness.
    \item We give three flow-based ways to instantiate the verifier: one using an
    almost-linear-time flow algorithm, one using specialised bipartite highest-label
    preflow--push, and one using the generic highest-label preflow--push algorithm
    available in NetworkX.
    \item We show how to compute the exact PJR$^+$ threshold, including the regime
    below $1$, by binary search using the above mincut routines; an appendix
    records a formulation using parametric mincuts.
\end{itemize}

\section{Preliminaries}
For a positive integer $d$, let $[d]=\{1,\ldots,d\}$. 
\begin{definition}[Approval-based election]
    An approval-based election
consists of a set $N=[n]$ of voters, a set $C$ of $m$ candidates, and an approval
ballot $A_i\subseteq C$ for every voter $i\in N$. A \emph{committee} is a set
$W\subseteq C$ of size $k$, where $1<k< m$. For each candidate $c\in C$,
write $N_c=\{i\in N:c\in A_i\}$ for the set of voters who approve $c$.
\end{definition} 
\begin{definition}[$\alpha$-PJR$^+$]
For a
group $N'\subseteq N$, write
\[
    \Gamma_W(N')=W\cap\bigcup_{i\in N'}A_i\qquad\text{and} \qquad r_W(N')=|\Gamma_W(N')|
\]
so \(\Gamma_W(N')\) is the set of committee members approved by at least one
voter in \(N'\), and \(r_W(N')\) is its
\emph{representation} in $W$.
Fix $\alpha>0$ and a committee $W$. A \emph{$\alpha$-PJR$^+$ violation} is a triple
$(c,N',\ell)$ such that $c\in C\setminus W$, $\emptyset\neq N'\subseteq N_c$, and
$\ell\in[k]$, with
\(
    |N'|\geq \alpha\ell n/k
    \text{ and }
    r_W(N')<\ell.
\)
The committee $W$ satisfies \emph{$\alpha$-PJR$^+$} if no such triple exists.
\end{definition}

The case $\alpha=1$ is ordinary PJR$^+$
\citep{DBLP:conf/sigecom/Brill023}. The definition for \(\alpha\geq1\) appeared
in our earlier work \citep{springham26a}; here we extend it to \(0<\alpha<1\).
Increasing \(\alpha\) raises the minimum size of a violating group and therefore
weakens the requirement.

\begin{definition}[Directed cuts]
    Let $G=(V,E,u)$ be a directed graph with nonnegative edge capacities
$u:E\to\mathbb{R}_{\geq0}$ and distinct vertices $s,z\in V$. For a set
$S\subseteq V$ with $s\in S$ and $z\notin S$, the capacity of the corresponding
$s$--$z$ cut is
\[
    \operatorname{cap}_G(S)
    =
    \sum_{\substack{(x,y)\in E\\x\in S,\ y\notin S}}u(x,y).
\]
A \emph{minimum $s$--$z$ cut} is a set $S$ of minimum capacity among
all such sets.
\end{definition}
Only edges directed from $S$ to $V\setminus S$ are counted; edges directed into
$S$ are not. 

\section{The mincut formulation}
\label{sec:mincut}

This section turns the search for an \(\alpha\)-PJR$^+$ violation into a
mincut problem. For each unelected candidate \(c\), we build a directed
graph from voters who approve \(c\) and the committee members in \(W\).
Certain cuts in this graph encode voter groups: an excluded voter contributes
to the cut through its source edge, while a committee member approved by the
selected group contributes through its sink edge. We first derive the required
edge capacities from the violation condition, then define the graph and
prove the correspondence between its mincuts and voter groups.

Fix an unelected candidate \(c\). A group \(N'\subseteq N_c\) can claim any
representation level \(\ell\in[k]\) satisfying
\(\ell\leq k|N'|/(\alpha n)\); its largest such claim is
\[
    \ell^*(N')=\min\left\{k,\floor*{\frac{k|N'|}{\alpha n}}\right\}.
\]
It witnesses a violation of $\alpha$-PJR$^+$ exactly when
\(r_W(N')<\ell^*(N')\). For a group with \(r_W(N')<k\), this is equivalent to
\[
    r_W(N')+1\leq \frac{k}{\alpha n}|N'|.
\]
We therefore need to compare the number of voters in \(N'\) with the number of
committee members approved by at least one of them.

The verification algorithms in the next section consider values of $\alpha$ represented as
\(
    \alpha=jk/(nh),
\)
where \(j\) and \(h\) are positive integers. They represent a possible group
size and representation level, respectively: at the breakpoint induced by a
group \(N'\) and level \(\ell\), we have \(j=|N'|\) and \(h=\ell\). For such a value,
\(h/j=k/(\alpha n)\), and the preceding inequality is equivalent to
\[
    j\,r_W(N')-h|N'|\leq-j.
\]
Let
\(U\) denote the fixed set of voter vertices included in the graph. In the next section, we take \(U=N_c\)
when \(\alpha\geq1\), and \(U=N_c\setminus N_w\) for a chosen \(w\in W\) when
\(0<\alpha<1\). A cut
will select a group \(N'\subseteq U\).

\begin{definition}[Cut graph]
\label{def:cut-network}
For a voter set \(U\subseteq N\) and positive integers \(j\) and \(h\), define
\(G(U;j,h)\) to have vertex set \(\{s,z\}\cup U\cup W\), where \(s\) is the
source and \(z\) is the sink, and the following edges:
\begin{align*}
    &(s,i) &&\text{of capacity \(h\), for every \(i\in U\)},\\
    &(i,w) &&\text{of capacity \(h\), for \(i\in U\) and
        \(w\in A_i\cap W\)},\\
    &(w,z) &&\text{of capacity \(j\), for every \(w\in W\)}.
\end{align*}
\end{definition}
\Cref{fig:cut-network} shows an example with three voters and three committee
members. The voter--winner edges encode approvals. The underlying undirected
graph is bipartite, with parts \(\{s\}\cup W\) and \(U\cup\{z\}\).

If \(N'\) and all winners in \(\Gamma_W(N')\) are placed on the source side,
the source edges of the voters in \(U\setminus N'\) contribute
\(h(|U|-|N'|)\) in total to the value of the cut, and the
sink edges of the committee members in \(\Gamma_W(N')\) contribute
\(j\,r_W(N')\) in total. This can be seen from \Cref{fig:canonical-cut}.
Additionally, unlike in the standard maximum-closure reduction, the voter--winner edges have
finite capacity \(h\). This is sufficient because, whenever a cut contains an
edge from a source-side voter to a sink-side winner, moving that voter to the
sink side adds the capacity-\(h\) source edge but removes at least one
capacity-\(h\) voter--winner edge. The cut capacity therefore does not
increase. The canonical-cut lemma below formalises both of these observations.

\begin{figure}[t]
    \centering
    \begin{tikzpicture}
  \node[style=terminal] (s) at (0,0) {$s$};
  \node[style=voter] (i1) at (2,2) {$1$};
  \node[style=voter] (i2) at (2,0) {$2$};
  \node[style=voter] (i3) at (2,-2) {$3$};
  \node[style=winner] (w1) at (5,2) {$w_1$};
  \node[style=winner] (w2) at (5,0) {$w_2$};
  \node[style=winner] (w3) at (5,-2) {$w_3$};
  \node[style=terminal] (z) at (7,0) {$z$};

  \draw[style=flow edge] (s) to node[style=edge label,above] {$h$} (i1);
  \draw[style=flow edge] (s) to node[style=edge label,above] {$h$} (i2);
  \draw[style=flow edge] (s) to node[style=edge label,below] {$h$} (i3);
  \draw[style=flow edge] (i1) to node[style=edge label,above] {$h$} (w1);
  \draw[style=flow edge] (i1) to node[style=edge label,above] {$h$} (w2);
  \draw[style=flow edge] (i2) to node[style=edge label,above] {$h$} (w2);
  \draw[style=flow edge] (i3) to node[style=edge label,below] {$h$} (w2);
  \draw[style=flow edge] (i3) to node[style=edge label,below] {$h$} (w3);
  \draw[style=flow edge] (w1) to node[style=edge label,above] {$j$} (z);
  \draw[style=flow edge] (w2) to node[style=edge label,above] {$j$} (z);
  \draw[style=flow edge] (w3) to node[style=edge label,below] {$j$} (z);
\end{tikzpicture}
    \caption{An example of \(G(U;j,h)\). Here
    \(U=\{1,2,3\}\) and \(W=\{w_1,w_2,w_3\}\).}
    \label{fig:cut-network}
\end{figure}
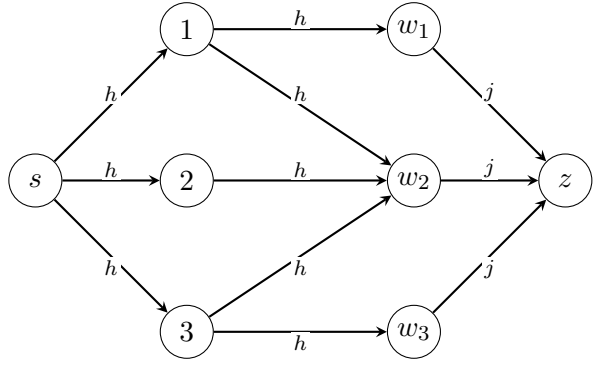

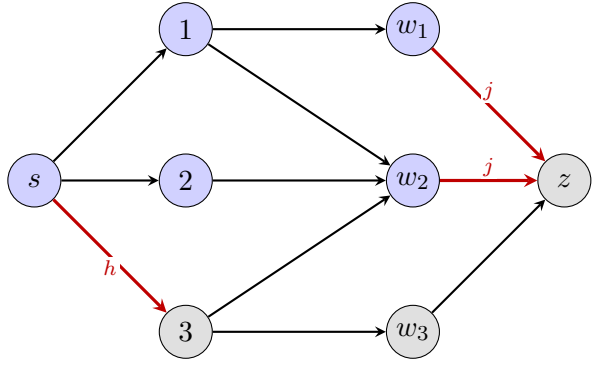
\begin{figure}[t]
    \centering
    \begin{tikzpicture}
  \node[style=source side] (s) at (0,0) {$s$};
  \node[style=source side] (i1) at (2,2) {$1$};
  \node[style=source side] (i2) at (2,0) {$2$};
  \node[style=sink side] (i3) at (2,-2) {$3$};
  \node[style=source side] (w1) at (5,2) {$w_1$};
  \node[style=source side] (w2) at (5,0) {$w_2$};
  \node[style=sink side] (w3) at (5,-2) {$w_3$};
  \node[style=sink side] (z) at (7,0) {$z$};

  \draw[style=flow edge] (s) to (i1);
  \draw[style=flow edge] (s) to (i2);
  \draw[style=cut edge] (s) to node[style=edge label,below] {$h$} (i3);
  \draw[style=flow edge] (i1) to (w1);
  \draw[style=flow edge] (i1) to (w2);
  \draw[style=flow edge] (i2) to (w2);
  \draw[style=flow edge] (i3) to (w2);
  \draw[style=flow edge] (i3) to (w3);
  \draw[style=cut edge] (w1) to node[style=edge label,above] {$j$} (z);
  \draw[style=cut edge] (w2) to node[style=edge label,above] {$j$} (z);
  \draw[style=flow edge] (w3) to (z);
\end{tikzpicture}
    \caption{The canonical cut for \(N'=\{1,2\}\) in the example from
    \Cref{fig:cut-network}. Its value is
    \(h|U\setminus N'|+j|\Gamma_W(N')|=h+2j\).}
    \label{fig:canonical-cut}
\end{figure}

\begin{lemma}[Canonical cuts]
\label{lem:canonical}
Every \(s\)--\(z\) cut in \(G(U;j,h)\) can be transformed, without increasing
its capacity, into one with source side
\(
 S_{N'}=\{s\}\cup N'\cup\Gamma_W(N')
\)
for some \(N'\subseteq U\).
\end{lemma}

\begin{proof}
Start from an arbitrary source side \(S\). Move to the sink side every voter
\(i\in U\cap S\) that approves some winner outside \(S\). Moving \(i\) adds
the capacity-\(h\) edge \((s,i)\) to the cut but removes at least one
capacity-\(h\) voter--winner edge. Every other affected voter--winner edge
either ceases to cross the cut or points into the new source side. The cut
capacity therefore does not increase.

Let \(N'\) be the voters that remain on the source side. Every winner in
\(\Gamma_W(N')\) is now on the source side. Move every other source-side winner
to the sink side. This removes its capacity-\(j\) sink edge and creates no
crossing voter--winner edge. The resulting source side is exactly \(S_{N'}\), and
its capacity is no larger than that of the original cut.
\end{proof}

\Cref{fig:canonical-cut} illustrates the resulting canonical cut. Blue
vertices lie on its source side and grey vertices on its sink side;
the red edges are exactly those counted in its capacity.

\begin{proposition}[Cut-value characterisation]
\label{prop:cut-value}
Writing \(M(U;j,h)\) for the mincut value,
\[
 M(U;j,h)=\min_{N'\subseteq U}\bigl(h(|U|-|N'|)+j\,r_W(N')\bigr).
\]
Consequently,
\[
 M(U;j,h)\leq h|U|-j
 \quad\Longleftrightarrow\quad
 \exists N'\subseteq U:\quad r_W(N')-\frac{h}{j}|N'|\leq-1.
\]
\end{proposition}

\begin{proof}
By \Cref{lem:canonical}, some mincut has source side \(S_{N'}\). Its
crossing edges are precisely the source edges for \(U\setminus N'\) and the
sink edges for \(\Gamma_W(N')\), giving the formula. The equivalence follows by
subtracting \(h|U|\) and dividing by \(j>0\).
\end{proof}

A returned minimum \(s\)--\(z\) cut need not already have canonical form
because the capacity-\(h\) middle edges can create ties. Given its source side
\(S\), form
\[
    N'=\{i\in U\cap S:A_i\cap W\subseteq S\}
\]
by scanning the voter--winner edges, and then form \(\Gamma_W(N')\) from the
edges incident to \(N'\). Replacing \(S\) by
\(\{s\}\cup N'\cup\Gamma_W(N')\) gives another minimum \(s\)--\(z\) cut,
now in canonical form. If \(e_U\) is
the number of voter--winner edges, this takes \(\Oh(|U|+k+e_U)\) time. Thus
recovering \(N'\) adds only linear overhead to the flow computation.

We have therefore reduced the coverage-minus-linear objective to one minimum
cut and shown how to recover a group attaining its minimum. The next section
chooses \(U\), \(j\), and \(h\) so that this group is exactly a
\(\alpha\)-PJR$^+$ violation witness.

\section{Verifying \texorpdfstring{\(\alpha\)-PJR$^+$}{alpha-PJR+}}
\label{sec:verification}

We now use the cut formulation to search for a violating triple
\((c,N',\ell)\). For each unelected candidate \(c\), the mincut selects a
group from voters who approve \(c\). When \(\alpha\geq1\), we may use all of
\(N_c\). When \(0<\alpha<1\), we must additionally ensure that some committee
member is approved by no voter in the selected group. This is the reason for
the two algorithms below.

The algorithms first replace an arbitrary query \(\alpha>0\) by a value of the form
\(jk/(nh)\), with \(j\in[n]\) and \(h\in[k]\). This does not change the answer.
Indeed, fix \(c\in C\setminus W\), a nonempty \(N'\subseteq N_c\), and
\(\ell\in[k]\) with \(r_W(N')<\ell\). The triple \((c,N',\ell)\) is a
violation exactly when
\(
 \alpha\leq k|N'|/n\ell,
\)
and every value on the right belongs to
\[
 \mathcal R(n,k)=
 \left\{\frac{jk}{nh}:j\in[n],\ h\in[k]\right\}.
\]
Let \(\alpha'\) be the smallest member of
\(\mathcal R(n,k)\) with \(\alpha'\geq \alpha\). The committee fails
\(\alpha\)-PJR$^+$ if and only if it fails \(\alpha'\)-PJR$^+$. If no such
value exists, it satisfies \(\alpha\)-PJR$^+$. Otherwise, we replace \(\alpha\)
by \(\alpha'\) and continue to denote the rounded value by \(\alpha\). Thus,
for the remainder of this section, \(\alpha=jk/(nh)\). If a query below \(1\)
is rounded to \(1\), we use the \(\alpha\geq1\) algorithm.

\subsection{The case \texorpdfstring{\(\alpha\geq1\)}{alpha >= 1}}

Here we take \(U=N_c\). If the cut condition from
\Cref{prop:cut-value} holds, then
\(r_W(N')+1\leq h|N'|/j\). The assumption \(\alpha\geq1\) ensures that
\(h|N'|/j\leq k\), so the representation level
\(\ell=\floor*{h|N'|/j}\) lies in \([k]\).

\begin{algorithm}[t]
\caption{Find an \(\alpha\)-PJR$^+$ violation for \(\alpha\geq1\)}
\label{alg:large-t}
\KwIn{Election, committee \(W\), and \(\alpha=jk/(nh)\geq1\).}
\KwOut{A violating triple, or \(\textsf{None}\).}
\ForEach{\(c\in C\setminus W\)}{
 \(U\gets N_c\)\;
 \lIf{\(j>h|U|\)}{\textbf{continue}}
 Compute a mincut of \(G(U;j,h)\), with value \(M\) and source side \(S\)\;
 \If{\(M\leq h|U|-j\)}{
  Canonicalise \(S\) to obtain \(N'\)\;
  \(\ell\gets\floor*{h|N'|/j}\)\;
  \Return{\((c,N',\ell)\)}\;
 }
}
\Return{\(\textsf{None}\)}\;
\end{algorithm}

\begin{theorem}[Correctness for \(\alpha\geq1\)]
\label{thm:large-correct}
\Cref{alg:large-t} returns an \(\alpha\)-PJR$^+$ violation if one exists, and returns
\(\textsf{None}\) otherwise.
\end{theorem}

\begin{proof}
If the cut test succeeds, \Cref{prop:cut-value} gives
\(
 r_W(N')+1\leq h|N'|/j.
\)
The left side is an integer, so
\(r_W(N')+1\leq\floor*{h|N'|/j}=\ell\). Thus \(\ell\geq1\) and
\(N'\neq\emptyset\). Moreover,
\(
 h|N'|/j=k|N'|/(\alpha n)\leq k,
\)
using \(\alpha\geq1\) and \(|N'|\leq n\). Hence \(\ell\in[k]\),
\(r_W(N')<\ell\). Moreover,
\(\ell\leq h|N'|/j\) is equivalent to
\(|N'|\geq \alpha\ell n/k\). Since \(N'\subseteq U=N_c\), the returned triple is
an \(\alpha\)-PJR$^+$ violation.

Conversely, a violation \((c,N',\ell)\) satisfies
\(h|N'|/j\geq\ell\) and \(r_W(N')\leq\ell-1\). Therefore its graph has a canonical cut with $N'$, and the canonical cut associated with $N'$ passes the cut test by \Cref{prop:cut-value}. The graph is not skipped because
\(j\leq h|N'|\leq h|U|\).
\end{proof}

\subsection{The case \texorpdfstring{\(0<\alpha<1\)}{0 < alpha < 1}}

When \(\alpha<1\), \(\floor*{k|N'|/(\alpha n)}\) may exceed \(k\), so it need not be a
valid representation level. For example, let
\(\alpha=1/2\), \(C=W\cup\{c\}\), and suppose every voter approves every candidate.
Then every nonempty \(N'\subseteq N_c=N\) has \(r_W(N')=k\), so it cannot be a
violation. Nevertheless, the unrestricted objective at \(N'=N\) is
\[
 r_W(N)-\frac{k}{\alpha n}|N|=k-2k=-k\leq-1.
\]
Thus the unrestricted cut test would give a false positive. For every
\(w\in W\), however, \(N_c\setminus N_w=\emptyset\), so the restricted
construction below correctly finds no group.

The restriction follows from the definition of a violation. If
\((c,N',\ell)\) is a violation, then \(r_W(N')<\ell\leq k\), so some
\(w\in W\) is approved by no voter in \(N'\). The algorithm therefore tries
each \(w\in W\) and uses \(U=N_c\setminus N_w\). Conversely, every group
selected from this \(U\) satisfies \(r_W(N')\leq k-1\), allowing its
representation level to be capped at \(k\).

\begin{algorithm}[t]
\caption{Find an \(\alpha\)-PJR$^+$ violation for \(0<\alpha<1\)}
\label{alg:small-t}
\KwIn{Election, committee \(W\), and \(\alpha=jk/(nh)<1\).}
\KwOut{A violating triple, or \(\textsf{None}\).}
\ForEach{\(c\in C\setminus W\)}{
 \ForEach{\(w\in W\)}{
  \(U\gets N_c\setminus N_w\)\;
  \lIf{\(j>h|U|\)}{\textbf{continue}}
  Compute a mincut of \(G(U;j,h)\), with value \(M\) and source side \(S\)\;
  \If{\(M\leq h|U|-j\)}{
   Canonicalise \(S\) to obtain \(N'\)\;
   \(\ell\gets\min\{k,\floor*{h|N'|/j}\}\)\;
   \Return{\((c,N',\ell)\)}\;
  }
 }
}
\Return{\(\textsf{None}\)}\;
\end{algorithm}

\begin{theorem}[Correctness for \(0<\alpha<1\)]
\label{thm:small-correct}
\Cref{alg:small-t} returns an \(\alpha\)-PJR$^+$ violation if one exists, and returns
\(\textsf{None}\) otherwise.
\end{theorem}

\begin{proof}
If the cut test succeeds, \Cref{prop:cut-value} gives
\(r_W(N')+1\leq h|N'|/j\). Since the left side is an integer,
\(
 r_W(N')+1\leq\floor*{h|N'|/j}.
\)
Also \(N'\subseteq U=N_c\setminus N_w\), so \(r_W(N')\leq k-1\). If the
floor is at most \(k\), it equals \(\ell\); otherwise \(\ell=k\). In either
case \(r_W(N')<\ell\), and \(\ell\in[k]\). Finally,
\(\ell\leq h|N'|/j=k|N'|/(\alpha n)\), which is equivalent to
\(|N'|\geq \alpha\ell n/k\). Thus the returned triple is a violation.

Conversely, if \((c,N',\ell)\) is a violation, then some \(w\in W\) is approved
by no voter in \(N'\), because \(r_W(N')<\ell\leq k\). Hence
\(N'\subseteq N_c\setminus N_w\). The inequalities
\(h|N'|/j\geq\ell\) and \(r_W(N')\leq\ell-1\) make this graph pass the cut
test. It is not skipped because \(j\leq h|N'|\leq h|U|\).
\end{proof}

\subsection{Running time and implementation}

For one voter set \(U\), the graph has
\[
 a=|U|+k+2\quad\text{vertices},\qquad
 b=|U|+k+e_U\quad\text{edges},
\]
where \(e_U\) is the number of voter--winner edges:
\[
 e_U=\bigl|\{(i,w)\in U\times W:w\in A_i\}\bigr|.
\]
Thus
\(a=\Oh(n+k)\), \(b=\Oh(nk)\), and capacities are at most
\(\max\{n,k\}\). Construction, residual reachability, and canonicalisation take
\(\Oh(b)\) time.

\begin{proposition}[Verification running times]
\label{prop:verification-runtime}
For \(\alpha\geq1\), a violation can be found, or its nonexistence certified, in
\(\Oh(m(nk)^{1+o(1)})\) time using almost-linear maximum flow,
\(\Oh(mnk\min\{n,k\})\) time using specialised bipartite highest-label
preflow--push, or
\(\Oh(m(n+k)^2\sqrt{nk})\) time using generic highest-label preflow--push.
For \(0<\alpha<1\), each bound gains a factor of \(k\).
\end{proposition}

\begin{proof}
The two algorithms construct \(\Oh(m)\) and \(\Oh(mk)\) graphs,
respectively. Almost-linear maximum flow takes \(b^{1+o(1)}\) time here
\citep{DBLP:conf/focs/Brand0PKLGSS23,DBLP:journals/jacm/ChenKLPGS25}.

For the specialised bipartite method, let \(p=\min\{|U|+1,k+1\}\), the smaller
bipartition size. Its highest-label bound is
\[
 \Oh\bigl(pb+\min\{p^3,p^2\sqrt b\}\bigr)
\]
per graph \citep{DBLP:journals/siamcomp/AhujaOST94}. Substituting
\(b=\Oh(nk)\) and \(p=\Oh(\min\{n,k\})\), and observing that
\(p^3=\Oh(pnk)\), gives the stated coarse bound.
Generic highest-label preflow--push takes \(\Oh(a^2\sqrt b)\) time per
graph \citep{DBLP:journals/siamcomp/CheriyanM89}. Multiplying by the number
of graphs proves the claim.
\end{proof}

\begin{remark}[Immediate implementation]
The graphs can be passed directly to NetworkX's highest-label
\texttt{preflow\_push} routine \citep{hagberg2008}. Residual reachability from
\(s\) gives a mincut source side, after which canonicalisation extracts
\(N'\). This uses the generic, not the specialised, bound above.
\end{remark}

We can therefore decide \(\alpha\)-PJR$^+$ and, when it fails, return an explicit
violating triple. The next section uses this verifier to locate the exact
boundary between values of \(\alpha\) that fail and values that hold.

\section{Computing the exact \texorpdfstring{PJR$^+$}{PJR+} threshold}
\label{sec:threshold}

Section~\ref{sec:verification} decides whether a committee satisfies
\(\alpha\)-PJR$^+$ for a fixed \(\alpha\). This section instead asks for the exact point
at which the answer changes. If the committee fails at some value \(\alpha\), then
it also fails at every smaller positive value, because decreasing \(\alpha\) only
reduces the required group size. There is therefore a boundary \(\tau(W)\):
the committee fails for \(0<\alpha\leq\tau(W)\) and satisfies the axiom for
\(\alpha>\tau(W)\). The boundary itself is included in the failing range. A smaller
value of \(\tau(W)\) is a stronger guarantee, since the committee then satisfies
PJR$^+$ for a larger range of parameters. We compute this boundary by
binary search using the verifier from the previous section. A specialised
parametric formulation is included in Appendix~\ref{app:parametric} for
completeness.

To derive the boundary, fix an unelected candidate \(c\) and a nonempty group
\(N'\subseteq N_c\) with \(r_W(N')<k\). The smallest representation level at
which this group can be underrepresented is \(\ell=r_W(N')+1\). At this level,
the group is large enough to violate the axiom precisely when
\[
 \alpha\leq \frac{k|N'|}{n(r_W(N')+1)}.
\]
Thus the ratio on the right is the largest parameter at which this particular
pair \((c,N')\) can cause a violation. The threshold is the largest such ratio
over all possible pairs.

\begin{definition}[PJR$^+$ threshold]
For each \(c\in C\setminus W\) and nonempty \(N'\subseteq N_c\) with
\(r_W(N')<k\), define
\[
 q(c,N')=\frac{k|N'|}{n(r_W(N')+1)}.
\]
The \emph{PJR$^+$ threshold} of \(W\) is
\[
 \tau(W)=
 \max_{\substack{c\in C\setminus W,\ \emptyset\neq N'\subseteq N_c\\r_W(N')<k}}
 q(c,N'),
\]
where the maximum of the empty set is \(0\).
\end{definition}

When \(\tau(W)>0\), a \emph{boundary witness} is a triple
\((c,N',\ell)\) that is a violation at \(\alpha=\tau(W)\). Such a witness records
both the group that determines the threshold and the representation level it
claims.

\begin{proposition}[Boundary semantics]
\label{prop:threshold-boundary}
For \(\alpha>0\), \(W\) fails \(\alpha\)-PJR$^+$ exactly when
\(\alpha\leq\tau(W)\). If
\(\tau(W)>0\), a maximizing pair \((c,N')\) gives the boundary witness
\((c,N',r_W(N')+1)\). Hence ordinary PJR\(^+\) holds exactly when
\(\tau(W)<1\).
\end{proposition}

\subsection{Binary search with ordinary mincuts}

Recall the finite set of possible query values
\[
 \mathcal R(n,k)=
 \left\{\frac{jk}{nh}:j\in[n],\ h\in[k]\right\}.
\]
Every positive threshold belongs to this set: for a pair \((c,N')\) attaining
the maximum in the definition, take \(j=|N'|\) and
\(h=r_W(N')+1\). We can therefore find the threshold by binary search over
\(\mathcal R(n,k)\), using only the ordinary mincut verifiers from
Section~\ref{sec:verification}.

\begin{algorithm}[t]
\caption{Compute the PJR$^+$ threshold by binary search}
\label{alg:threshold}
\KwIn{Election and committee \(W\).}
\KwOut{\((\tau(W),Q)\), where \(Q\) is a triple that violates
\(\tau(W)\)-PJR$^+$, or \(\textsf{None}\) if \(\tau(W)=0\).}
Form the sorted distinct list \(R=\mathcal R(n,k)\). With each value \(\alpha\),
store any pair \((j,h)\) such that \(\alpha=jk/(nh)\)\;
Run \Cref{alg:large-t} at \(\alpha=1\), represented by \(j=n,h=k\)\;
\eIf{a violation is returned}{
 \(R'\gets R\cap[1,\infty)\), \(\mathcal A\gets\Cref{alg:large-t}\)\;
}{
 \(R'\gets R\cap(0,1)\), \(\mathcal A\gets\Cref{alg:small-t}\)\;
}
Binary-search \(R'\) for its largest value \(\widehat\alpha\) at which \(\mathcal A\)
returns a violation\;
\lIf{no such value exists}{\Return{\((0,\textsf{None})\)}}
Run \(\mathcal A\) once more at \(\widehat\alpha\) to obtain \(Q\)\;
\Return{\((\widehat\alpha,Q)\)}\;
\end{algorithm}

\begin{theorem}[Correctness of binary search]
\label{thm:threshold-correct}
\Cref{alg:threshold} returns the threshold and, when it is positive, a
violation at the boundary.
\end{theorem}

\begin{proof}
By \Cref{prop:threshold-boundary}, failure occurs exactly for
\(0<\alpha\leq\tau(W)\), and every positive threshold belongs to
\(\mathcal R(n,k)\). If the query at \(1\) returns a violation, then
\(\tau(W)\geq1\);
the \(\alpha\geq1\) verifier (\Cref{alg:large-t}) therefore finds the largest failing value in
\(R\cap[1,\infty)\). If the query at \(1\) returns no violation, then
\(\tau(W)<1\);
the \(0<\alpha<1\) verifier (\Cref{alg:small-t}) finds the largest failing value in
\(R\cap(0,1)\). Thus binary search returns \(\tau(W)\). At this value, the
final verifier call returns a violation at the boundary. If no positive
breakpoint fails, the threshold is \(0\).
\end{proof}

\begin{proposition}[Binary-search running times]
\label{prop:threshold-runtime}
In the worst case, the running times of \Cref{alg:threshold} are
\[
\begin{array}{ll}
\Oh(mk(nk)^{1+o(1)})&
    \text{with almost-linear maximum flow},\\
\Oh(mnk^2\min\{n,k\}\log(nk))&
    \text{with specialised bipartite preflow--push},\\
\Oh(mk(n+k)^2\sqrt{nk}\log(nk))&
    \text{with generic preflow--push}.
\end{array}
\]
When \(\tau(W)\geq1\), the factor \(k\) can be removed from each bound.
\end{proposition}

\begin{proof}
There are at most \(nk\) breakpoints, so binary search makes
\(\Oh(\log(nk))\) verifier calls. If the query at \(1\) finds a violation,
all subsequent calls use the \(\alpha\geq1\) verifier. Otherwise they use the
\(0<\alpha<1\) verifier, whose bounds in \Cref{prop:verification-runtime} have one
additional factor of \(k\). For almost-linear flow, the binary-search factor is
absorbed because
\((nk)^{1+o(1)}\log(nk)=(nk)^{1+o(1)}\); it remains explicit in the two
push--relabel bounds. Constructing and sorting the breakpoints is dominated.
\end{proof}

Appendix~\ref{app:parametric} records a parametric formulation for completeness;
it has effectively the same running time as the specialised bipartite algorithm,
but may nevertheless be of interest.

\subsection{Comparison with \texorpdfstring{EJR$^+$}{EJR+}}

EJR$^+$, introduced by \citet{DBLP:conf/sigecom/Brill023}, provides a useful
comparison. Applying the same group-size scaling as above, we define an
\(\alpha\)-EJR$^+$ violation to consist of an unelected candidate \(c\), a level
\(\ell\in[k]\), and at least \(\alpha\ell n/k\) voters who approve \(c\) and each
approve fewer than \(\ell\) committee members. For fixed \(c\) and \(\ell\),
the largest possible such group is
\(
 B_{c,\ell}
 =
 \{i\in N_c:|A_i\cap W|<\ell\}.
\)
Consequently, \(W\) fails \(\alpha\)-EJR$^+$ exactly when
\(
 |B_{c,\ell}|\geq \alpha\ell n/k
\)
for some \(c\in C\setminus W\) and \(\ell\in[k]\). Iterating over \(c\),
\(\ell\), and the voters therefore verifies \(\alpha\)-EJR$^+$ in
\(\Oh(mnk)\) time for every \(\alpha>0\).

The same computation gives the exact EJR$^+$ threshold directly:
\[
 \tau_{\mathrm{EJR}^+}(W)
 =
 \max_{\substack{c\in C\setminus W\\ \ell\in[k]}}
 \frac{k|B_{c,\ell}|}{n\ell},
\]
which can also be evaluated in \(\Oh(mnk)\) time. For \(\alpha\geq1\), our
PJR$^+$ verifier runs in \(\Oh(m(nk)^{1+o(1)})\) time and therefore matches
the EJR$^+$ bound up to a subpolynomial factor. For \(0<\alpha<1\), and in the
worst case when computing the exact PJR$^+$ threshold, the PJR$^+$ bound has
one additional factor of \(k\).

\section{Conclusion}

We have shown that the coverage-minus-linear objective used to verify PJR$^+$
is a maximum-closure problem and hence can be represented by a mincut in a
bipartite graph. This yields an \(\Oh(m(nk)^{1+o(1)})\)-time verifier that
also returns an explicit violating group. The same reduction can instead be
implemented using specialised bipartite highest-label preflow--push, or
directly using NetworkX's generic preflow--push routine.

The exact PJR$^+$ threshold can be computed by binary search over its finitely
many possible values, using the same mincut verifiers. For thresholds
below \(1\), exactness is obtained by considering, for each unelected candidate
and committee member, only the voters who do not approve that committee member.
Appendix~\ref{app:parametric} records a parametric formulation for completeness;
it has effectively the same running time as the specialised bipartite algorithm,
but may nevertheless be of interest.

It would be interesting to see how well an implementation of our generic preflow-push verifier would compare to an implementation of the submodular-minimisation approach in terms of actual implemented speed, not just asymptotics.

\section*{Acknowledgments}
Drew Springham was supported by UK Research and Innovation [grant number EP/S023356/1], in the UKRI Centre for Doctoral Training in Safe and Trusted Artificial Intelligence (www.safeandtrustedai.org).
\section*{Usage of AI}
Throughout the research process, the author has used various AI models, primarily Claude Opus 5, and GPT-5.6 Sol.

\bibliography{refs}
\appendix
\section{A parametric-flow formulation for the threshold}
\label{app:parametric}

This appendix records a parametric formulation of the threshold problem, primarily for interest. It replaces the binary-search calls by specialised
parametric flow computations, but its asymptotic running time is effectively the same as for the specialised bipartite preflow-push algorithm; we replace the $\Oh(\log (nk))=\Oh(\log(\max\{n,k\}))$ binary search factor with $\Oh(\log(\min\{n,k\}))$.

An ordinary mincut algorithm takes a graph with fixed edge capacities and
returns an \(s\)--\(z\) mincut. In a parametric mincut problem, some
capacities are instead functions of a scalar parameter \(\lambda\), and the aim
is to determine how the mincut changes as \(\lambda\) varies. In the
setting used here, each varying capacity is a linear function of \(\lambda\),
so the capacity of any fixed cut is also linear in \(\lambda\). The mincut
value is therefore piecewise linear in $\lambda$: it follows the value of one cut until
another cut becomes cheaper. A \emph{breakpoint} is a value of \(\lambda\) at
which such a change occurs. The all-breakpoints algorithm of
\citet{DBLP:journals/siamcomp/GalloGT89} computes these
breakpoints. Write
\(0\leq\beta_1<\cdots<\beta_q\) for the breakpoints. For positive
\(\lambda\), the parameter intervals are
\((0,\beta_1)\), \((\beta_j,\beta_{j+1})\) for \(1\leq j<q\), and
\((\beta_q,\infty)\), with empty intervals omitted; at a breakpoint, multiple
cuts may tie. On each parameter interval one may choose a mincut source
side that remains constant throughout the interval. Because
our parameter-dependent sink capacities increase with \(\lambda\), these source
sides can be chosen to form a decreasing chain.\footnote{Technically, the standard parametric-flow convention uses non-increasing capacities into the sink, rather than non-decreasing capacities. The same result applies here by replacing each parameter-dependent sink capacity \(\lambda\) by \(n-\lambda\) on \(0\leq\lambda\leq n\), using the affine reparameterisation \(\lambda\mapsto n-\lambda\); the relevant ratios are at most \(n\), and the corresponding cuts are unchanged after reparameterisation.}
The breakpoint representation that stores, for each vertex, the breakpoint
at which it leaves the source side is described in
\citep{DBLP:journals/siamcomp/GalloGT89}; it uses
\(\Oh(a)\) space, where \(a\) is the number of vertices, and scanning the
stored values recovers a cut for any one interval in \(\Oh(a)\) time.
Fix \(U\subseteq N\). For a parameter \(\lambda\geq0\), let
\(H(U;\lambda)\) have vertex set
\(\{s,z,d\}\cup U\cup W\), with the following edges:
\begin{align*}
 (s,i) &\quad\text{of capacity \(1\), for every \(i\in U\)},\\
 (i,w) &\quad\text{of capacity \(1\), for \(i\in U\) and
     \(w\in A_i\cap W\)},\\
 (i,d) &\quad\text{of capacity \(1\), for every \(i\in U\)},\\
 (w,z) &\quad\text{of capacity \(\lambda\), for every \(w\in W\)},\\
 (d,z) &\quad\text{of capacity \(\lambda\)}.
\end{align*}
\Cref{fig:parametric-network} illustrates the construction.

\begin{figure}[t]
    \centering
    \begin{tikzpicture}[scale=1.2, transform shape]
  \node[style=terminal] (s) at (0,0) {$s$};
  \node[style=voter] (i1) at (2,2) {$1$};
  \node[style=voter] (i2) at (2,0) {$2$};
  \node[style=voter] (i3) at (2,-2) {$3$};
  \node[style=winner] (w1) at (5,2) {$w_1$};
  \node[style=winner] (w2) at (5,0) {$w_2$};
  \node[style=winner] (w3) at (5,-2) {$w_3$};
  \node[style=winner] (d) at (5,-4) {$d$};
  \node[style=terminal] (z) at (7,0) {$z$};

  \draw[style=flow edge] (s) to (i1);
  \draw[style=flow edge] (s) to (i2);
  \draw[style=flow edge] (s) to (i3);
  \draw[style=flow edge] (i1) to (w1);
  \draw[style=flow edge] (i1) to (w2);
  \draw[style=flow edge] (i2) to (w2);
  \draw[style=flow edge] (i3) to (w2);
  \draw[style=flow edge] (i3) to (w3);
  \draw[style=flow edge] (i1) to (d);
  \draw[style=flow edge] (i2) to (d);
  \draw[style=flow edge] (i3) to (d);
  \draw[style=flow edge] (w1) to node[style=edge label,above] {$\lambda$} (z);
  \draw[style=flow edge] (w2) to node[style=edge label,above] {$\lambda$} (z);
  \draw[style=flow edge] (w3) to node[style=edge label,below] {$\lambda$} (z);
  \draw[style=flow edge] (d) to node[style=edge label,below] {$\lambda$} (z);
\end{tikzpicture}
    \caption{A small example of \(H(U;\lambda)\). Unlabelled edges have
    capacity \(1\). The voter--winner edges encode approvals, and every voter
    is also joined to the dummy vertex \(d\).}
    \label{fig:parametric-network}
\end{figure}
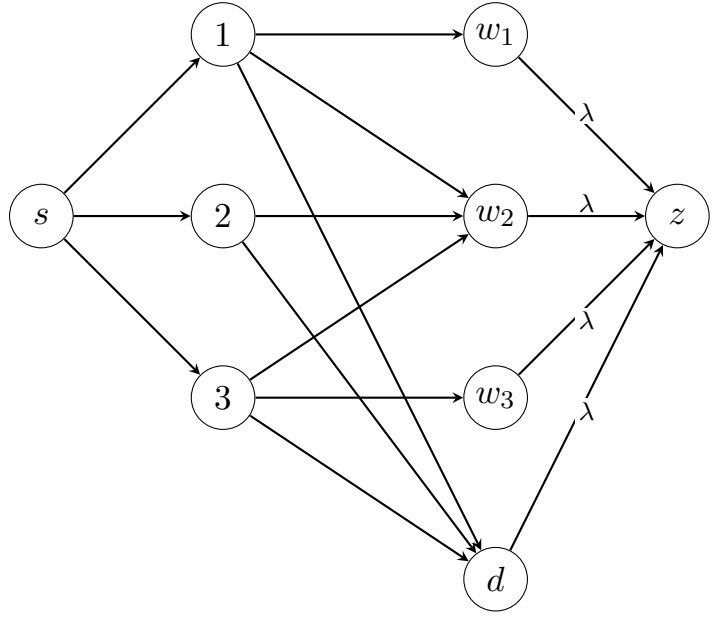

In our graph, only edges into the sink depend on \(\lambda\), and all of their
capacities increase linearly with \(\lambda\). We need to prove that the right
endpoint of the final parameter interval on which some mincut selects at
least one voter is
\[
  \max_{\emptyset\neq N'\subseteq U}
  \frac{|N'|}{r_W(N')+1},
\]
and that a cut on this interval identifies a group \(N'\) attaining the
maximum. The output of the parametric flow algorithm will then give both the exact value of this
maximum and a witnessing voter group in one computation.

For a nonempty group \(N'\subseteq U\), consider the cut with source side
\[
 \{s\}\cup N'\cup\Gamma_W(N')\cup\{d\}.
\]
We call this the canonical cut associated with \(N'\). The canonical cut
associated with the empty group has source side \(\{s\}\).

\begin{lemma}[Canonical cuts in the parametric graph]
\label{lem:parametric-canonical}
Every \(s\)--\(z\) cut in \(H(U;\lambda)\) can be transformed, without
increasing its capacity, into the canonical cut associated with some
\(N'\subseteq U\). Consequently, at least one global mincut has
canonical form.
\end{lemma}

\begin{proof}
Let \(S\) be the source side of an arbitrary cut. First suppose \(d\notin S\).
Move every source-side voter to the sink side. For each such voter \(i\), the
newly crossing edge \((s,i)\) has the same capacity as the previously crossing
edge \((i,d)\), while no other edge starts crossing from source to sink. Now
move every source-side committee member to the sink side. This can only remove
crossing edges, and the resulting source side is \(\{s\}\).

Now suppose \(d\in S\). As in the proof of \Cref{lem:canonical}, move to the
sink side every source-side voter that approves a committee member outside
\(S\). The new crossing source edge has capacity \(1\) and replaces at least
one crossing voter--winner edge of capacity \(1\), so the cut capacity does
not increase. Let \(N'\) be the voters that remain. Move every source-side
committee member outside \(\Gamma_W(N')\) to the sink side. This creates no
new crossing edge. If \(N'=\emptyset\), also move \(d\) to the sink side.
The resulting cut is the canonical cut associated with \(N'\).

Thus every cut has a canonical cut of no greater capacity. If the original
cut is globally minimum, the transformed cut cannot have strictly smaller
capacity, and hence is also globally minimum.
\end{proof}

For a nonempty \(N'\), the excluded voters contribute \(|U|-|N'|\) to its
canonical cut. The approved committee members contribute
\(\lambda r_W(N')\), and the dummy vertex contributes one further
\(\lambda\). Thus the cut value is
\(
 |U|-|N'|+\lambda(r_W(N')+1),
\)
whereas the canonical cut associated with the empty group has value \(|U|\).
The purpose
of \(d\) is now explicit: it creates the \(+1\) in the denominator below. The
cut for \(N'\) has value at most \(|U|\) exactly when
\[
 \lambda\leq\frac{|N'|}{r_W(N')+1}.
\]

\begin{proposition}[Parametric characterisation]
\label{prop:parametric}
Define
\[
 \lambda_U=\max_{\emptyset\neq N'\subseteq U}
 \frac{|N'|}{r_W(N')+1},
\]
with \(\lambda_{\emptyset}=0\). If \(U\neq\emptyset\), then \(\lambda_U\) is
the breakpoint separating two regimes: for \(\lambda>\lambda_U\), the empty
group has smaller cut value than every nonempty group; for
\(\lambda<\lambda_U\), the empty group is not optimal. Moreover, throughout
the interval immediately below \(\lambda_U\), every group represented by a
minimum \(s\)--\(z\) cut whose source side has canonical form attains the
maximum defining \(\lambda_U\).
\end{proposition}

\begin{proof}
The claim is immediate for \(U=\emptyset\). Suppose \(U\neq\emptyset\).
By \Cref{lem:parametric-canonical}, it suffices to compare canonical cuts. A nonempty group can have cut value at most
\(|U|\) exactly when
\(\lambda\leq |N'|/(r_W(N')+1)\). Taking the maximum over all nonempty groups
gives the two regimes.

There are only finitely many ratios
\(|N'|/(r_W(N')+1)\). Let \(\rho_U\) be the largest such ratio strictly below
\(\lambda_U\), taking \(\rho_U=0\) if there is none, and choose any real
\(\lambda\in(\rho_U,\lambda_U)\). The parameter \(\lambda\) may range over all
nonnegative real numbers and need not belong to \(\mathcal R(n,k)\). Each ratio
considered here is rational, and multiplying the maximizing ratio
\(\lambda_U\) by \(k/n\) later gives a threshold value in
\(\mathcal R(n,k)\). A group attaining \(\lambda_U\) then has cut
value below \(|U|\), so a mincut cannot represent the empty group. Any
nonempty group whose ratio is below \(\lambda_U\) has cut value above
\(|U|\), so it cannot be minimum either. Consequently every minimum
\(s\)--\(z\) cut at this value of \(\lambda\) whose source side has canonical
form represents a group attaining \(\lambda_U\).
\end{proof}

By \Cref{prop:parametric}, the last interval on which voters are selected ends
at \(\lambda_U\). At the breakpoint itself, the cut selecting no voters and
the cuts representing groups attaining \(\lambda_U\) all have value \(|U|\).
A mincut algorithm could therefore return the empty group. We instead use
a mincut from the interval immediately below \(\lambda_U\), whose
canonical form represents a group attaining \(\lambda_U\) by
\Cref{prop:parametric}.

\begin{theorem}[Exact threshold and running time by parametric flow]
\label{thm:parametric-threshold}
Run the verifier once at \(\alpha=1\), interpreting a maximum over an empty family
as \(0\).
\begin{enumerate}
    \item If it returns a violation, then
    \[
      \tau(W)=\frac{k}{n}\max_{c\in C\setminus W}\lambda_{N_c}.
    \]
    \item If it returns no violation, then
    \[
      \tau(W)=\frac{k}{n}
      \max_{\substack{c\in C\setminus W\\w\in W}}
      \lambda_{N_c\setminus N_w}.
    \]
\end{enumerate}
For each relevant set \(U\), the parametric algorithm may return several
breakpoints. Retain its final breakpoint \(\lambda_U\) and a canonical mincut from
the interval immediately below it. After choosing a set \(U\) for which
\(\lambda_U\) is largest, the voter set \(N'\) represented by this cut yields a
boundary witness by setting \(\ell=r_W(N')+1\).
The worst-case running time is
\[
 \Oh\bigl(mnk^2\min\{n,k\}\log(\min\{n,k\})\bigr).
\]
If the query at \(\alpha=1\) returns a violation, the factor \(k\) can be removed.
\end{theorem}

\begin{proof}
If the query at \(1\) returns a violation, then \(\tau(W)\geq1\). Maximizing
\(\lambda_{N_c}\) considers every nonempty group \(N'\subseteq N_c\), including
groups with \(r_W(N')=k\), which are excluded from the definition of
\(\tau(W)\). For any such excluded group, its value in the first displayed
maximum, after multiplication by \(k/n\), is at most
\[
 \frac{k}{n}\frac{|N'|}{k+1}\leq\frac{k}{k+1}<1,
\]
whereas \(\tau(W)\geq1\). The excluded groups therefore cannot change the
maximum. For every remaining group,
\[
 \frac{k}{n}\frac{|N'|}{r_W(N')+1}=q(c,N'),
\]
so the first formula is exactly the definition of \(\tau(W)\).

If the query at \(1\) returns no violation, then \(\tau(W)<1\), and we must
exclude groups with \(r_W(N')=k\) explicitly. For \(N'\subseteq N_c\),
\[
 r_W(N')<k
 \quad\Longleftrightarrow\quad
 \exists w\in W:\ N'\cap N_w=\emptyset
 \quad\Longleftrightarrow\quad
 \exists w\in W:\ N'\subseteq N_c\setminus N_w.
\]
Thus, as \(c\) and \(w\) vary, the sets \(N_c\setminus N_w\) contain exactly
the groups permitted in the definition of \(\tau(W)\). This proves the second
formula. In both cases, \Cref{prop:parametric} returns a group attaining the
relevant maximum, and \Cref{prop:threshold-boundary} turns it into a boundary
witness.

For the running time, one graph \(H(U;\lambda)\) has
\(b=\Oh(nk)\) edges and smaller bipartition size
\(p=\min\{|U|+1,k+2\}=\Oh(\min\{n,k\})\). After canonicalisation, among
groups with the same \(r_W(N')\), only one of maximum cardinality can be
optimal. Hence there are at most
\(\Oh(p)\) relevant breakpoints. The dynamic-tree
bipartite parametric-flow result of
\citet{DBLP:journals/siamcomp/AhujaOST94}
applies to these breakpoints and gives a running time of
\(\Oh(pb\log p)\) per voter set \(U\). If the query at \(1\)
returns a violation, there are \(\Oh(m)\) sets \(U=N_c\). Otherwise there are
\(\Oh(mk)\) sets \(U=N_c\setminus N_w\). Substitution gives the two stated
bounds.
\end{proof}

This method requires a specialised parametric-flow implementation; it is not
provided by NetworkX's ordinary \texttt{preflow\_push} routine.
%%%%%%%%%%%%%%%%%%%%%%%%%%%%%%%%%%%%%%%%%%%%%%%%%%%%%%%%%%%%%%%%%%%%%%%%%

% At the very end of the paper, please include your contact details:

\begin{contact}
Drew Springham\\
Informatics Department, King's College London\\
London, United Kingdom\\
\email{drew.springham@kcl.ac.uk}
\end{contact}

%%%%%%%%%%%%%%%%%%%%%%%%%%%%%%%%%%%%%%%%%%%%%%%%%%%%%%%%%%%%%%%%%%%%%%%%%

\end{document}